\documentclass{article}

\usepackage[T1]{fontenc}
\usepackage[utf8]{inputenc}
\usepackage{authblk}

\usepackage[letterpaper, hmargin=1.1in, top=1in, bottom=1.2in, footskip=0.6in]{geometry}

\usepackage{titlesec}
\titleformat{\subsection}[runin]{\normalfont\bfseries}{\thesubsection.}{.5em}{}[.]\titlespacing{\subsection}{0pt}{2ex plus .1ex minus .2ex}{.8em}
\titleformat{\subsubsection}[runin]{\normalfont\itshape}{\thesubsubsection.}{.3em}{}[.]\titlespacing{\subsubsection}{0pt}{1ex plus .1ex minus .2ex}{.5em}

\usepackage[labelfont=sc,font=small,labelsep=period]{caption}
\usepackage{amsmath} 
\usepackage{amssymb}
\usepackage{amsfonts}
\usepackage{latexsym}
\usepackage{amsthm}
\usepackage{amsxtra}
\usepackage{amscd}
\usepackage{bbm}
\usepackage{mathrsfs}
\usepackage{bm}
\usepackage{xcolor}

\usepackage{graphicx, color}

\definecolor{darkred}{rgb}{0.9,0,0.3}
\definecolor{darkblue}{rgb}{0,0.3,0.9}

\usepackage{ifthen}
\def\comment#1{\ifthenelse{\isodd{\value{page}}}{\marginpar{\raggedright\scriptsize{\textcolor{darkred}{#1}}}}{\marginpar{\raggedleft\scriptsize{\textcolor{darkred}{#1}}}}}  

\usepackage[nottoc,notlof,notlot]{tocbibind}
\usepackage{cite} 

\numberwithin{equation}{section}
\numberwithin{figure}{section}

\theoremstyle{plain} 
\newtheorem{theorem}{Theorem}[section]
\newtheorem*{theorem*}{Theorem}
\newtheorem{lemma}[theorem]{Lemma}
\newtheorem*{lemma*}{Lemma}

\newtheorem*{corollary*}{Corollary}

\newtheorem*{proposition*}{Proposition}

\newtheorem*{definition*}{Definition}

\newtheorem*{conjecture*}{Conjecture}

\theoremstyle{definition} 

\newtheorem*{example*}{Example}
\newtheorem{remark}[theorem]{Remark}
\newtheorem*{remark*}{Remark}

\renewcommand{\cal}{\mathcal}

\newcommand*{\deq}{\mathrel{\vcenter{\baselineskip0.65ex \lineskiplimit0pt \hbox{.}\hbox{.}}}=}

\renewcommand{\leq}{\leqslant}
\renewcommand{\geq}{\geqslant}
\renewcommand{\epsilon}{\varepsilon}

\DeclareMathOperator{\var}{Var}

\usepackage[T1]{fontenc}
\usepackage[utf8]{inputenc}
\usepackage{authblk}

\title{Convergence in High Probability of the Quantum Diffusion in a Random Band Matrix Model }
\author{Vlad Margarint}

\date{\today}

\begin{document}

\maketitle

\begin{abstract}
We consider Hermitian random band matrices $H$ in $d \geq 1 $ dimensions. The matrix elements $H_{xy},$ indexed by $x, y \in \Lambda \subset \mathbb{Z}^d,$ are independent, uniformly distributed random variable if $|x-y| $ is less than the band width $W,$ and zero otherwise. We update the previous results of the converge of quantum diffusion in a random band matrix model  from convergence of the expectation to convergence in high probability. The result is uniformly in the size $|\Lambda| $ of the matrix.

\end{abstract}

\begin{section}{Introduction}

Random band  matrices $H=\left(H_{xy}\right)_{x,y \in \Gamma}$ represent systems on a large finite graph with a metric. They are the natural intermediate models to study quantum propagation in disordered systems as they interpolate in between the Wigner matrices and Random Schr\"odinger operators. The elements $H_{xy}$ are independent random variables with variance $\sigma_{xy}^2=\mathbb{E}|H_{xy}|^2$ depending on the distance between the two sites. The variance decays with the distance on the scale $W$, called the band width of the matrix $H$. This terminology comes from the simplest model in which the graph is a path on $N$ vertices labelled by $\Gamma=\{1,2,\ldots, N\},$ and the matrix elements $H_{xy}$ are zero if $|x-y| \geq W.$ If $W=O(1)$ we obtain the one-dimensional Anderson type model (see \cite{4}) and if $W=N$ we recover the Wigner matrix. In the general Anderson model, introduced in \cite{4}, a random on-site potential $V$ is added to a deterministic Laplacian on a graph that is typically a regular box in $\mathbb{Z}^d\,.$ For higher dimensional models in which the graph is $\Gamma$ is a box in $\mathbb{Z}^d$, see \cite{5}. 

In \cite{1} it was proved that the quantum dynamics of $d$-dimensional band matrix is given by a superposition of heat kernels up to time scales $t \ll W^{d/3}\,.$ Note that diffusion is expected to hold for $t \sim W^{2}$ for $d=1$ and up to any time for $d \geq 3$ when the thermodynamic limit is taken. The threshold $d/3$ on the exponent is due to technical estimates on Feynman graphs.

The approach of this paper is similar with the one in [1]\,.  We normalize the entries of the matrix such that the rate of quantum jumps is of order one. In contrast with \cite{1}  in this paper double-rooted  Feynman graphs are used to estimate the variance of the quantum diffusion. The main result of this paper is upgrading the previous results on the convergence of expectation of the quantum diffusion from \cite{1} to convergence in high probability. 
\end{section}

\textbf{Acknowledgements} This result is obtained in ETH Z\"urich under the supervision of Prof. Dr. Antti Knowles. The author is grateful to Prof. Dr. Antti Knowles for the careful guiding into understanding the problem.

\begin{section}{Model and main result}

Let $\mathbb{Z}^d$ be the infinite lattice  with the Euclidean norm $|\cdot |_{\mathbb{Z}^d}$ and let $M$ be the number of points situated at distance at most $W$ ($W$ $\geq 2$) from the origin, i.e.\ $$M=M(W)=|\{ x \in \mathbb{Z}^d: 1\leq |\cdot |_{\mathbb{Z}^d} \leq W\}|\,.$$
For simplicity, we avoid working directly on an infinite lattice. Throughout our proof, we consider a $d$-dimensional finite periodic lattice $\Lambda_N \subset\mathbb{Z}^d $ ($d\geq 2$) of linear size $N$ equipped with the Euclidean norm $|\cdot |_{\mathbb{Z}^d}$. Specifically, we take $\Lambda_N$ to be a cube centered around the origin with the side length $N$, i.e.

$$\Lambda_N:=( [-N/2,N/2) \cap \mathbb{Z})^d\,.$$ 
We regard $\Lambda_N$ periodic, i.e.\ we equip it with the periodic addition and periodic distance

$$|x|:=\text{inf}\{|x+N\nu|_{\mathbb{\mathbb{Z}^d}}: \nu \in {Z}^d \}\,.$$
We analyze random matrices $H $ with band width $W$ and with elements $H_{xy}$, where $x$ and $y$ are indices of points in $\Lambda_N$\,. For introducing $H$ we first define a matrix 
$$S_{xy}:=\frac{\bold{1}(1\leq|x-y|\leq W)}{(M-1)}\,.$$ We consider $A=A^*=(A_{xy})$ a Hermitian random matrix whose upper triangular entries ($A_{xy}:x \leq y$) are independent random variables uniformly distributed on the unit circle $\mathbb{S}^1 \subset \mathbb{C}$\,. We define the \emph{random band matrix} $(H_{xy})$ through
\begin{align*}
 H_{xy}&:=\sqrt{S_{xy}}A_{xy}\,.
\end{align*}
Note that $H$ is Hermitian and $|H_{xy}|^{2}=S_{xy}$\,. \\
Throughout our investigation we will use the simplified notation $\sum\limits_{y_1}$  for$\sum\limits_{y_1 \in \Lambda_N}$.\\
Our main quantity is
$$ P(t,x)=|(e^{-itH/2})_{0x}|^{2}\,.$$
The function $P(t,x)$ describes the quantum transition probability of a particle starting in $x_0$ and ending up at position $x$ after time $t$\,.\\ 
Let $\kappa>0$\,. We introduce the macroscopic time and space coordinates $T$ and $X$, which are independent of $W$, and consider the microscopic time and space coordinates

\begin{align}
t\;=\; &W^{d\kappa}T\,, \nonumber\\
x\;=\; &W^{1+d \kappa/2}X\,.\nonumber
\end{align}

Using the definition of the quantum probability and the scaling that we have introduced before, we define the random variable that we are going to investigate by

\begin{equation} 
Y_{T,\kappa, W}(\phi)\equiv Y_T(\phi):=\sum_{x}P(W^{d\kappa}T, x)\phi\left(\frac{x}{W^{1+d\kappa/2}}\right)\,,
\end{equation}
where $\phi \in C_b(\mathbb{R}^d)$ is a test function in $\mathbb{R}^d$\,.\\
Our main result gives an estimate for the variance of the random variable $Y_T(\phi)$ up to time scales $t=O(W^{d\kappa})$ if $\kappa <1/3$\,.

\begin{theorem}

Fix $T_0>0$ and $\kappa$ such that $0< \kappa <1/3\,.$ Choose a real number $\beta$ satisfying $0<\beta<2/3-2\kappa\,.$ Then there exists $C$ $\geq 0$ and $W_0 \geq 0$ depending only on $T_0$, $\kappa$ and $\beta$ such that for all $T \in [0,T_0]\,,$ $W \geq W_0$ and $N \geq W^{1+\frac{d}{6}}$ we have

$$\var(Y_T(\phi)) \leq \frac{C ||\phi||^{2}_\infty}{W^{d\beta}}\,. $$ 
\end{theorem}

\begin{remark}
Using the estimate that we obtain in Theorem $2.1$ and Chebyshev inequality for the second moment we obtain the convergence in high probability of the random variable $Y_{T}(\phi)\,.$ We think that the same technique can be implemented for a graphical representation with $2p$ directed chains with $p \in \mathbb{N}\,.$  This approach should give similar estimates on the $2p$-th moment of our random variable that we further use in the Chebyshev's inequality to get the desired conclusion.
\end{remark}

\end{section}

\begin{section}{Graphical representation}

In this section we give the exact formula of the  quantity of our analysis and we motivate the graphical representation that we will use in order to compute the upper bound.

\begin{subsection}{Expansion in non-backtracking powers}

First, as in \cite{1} we define $H^{(n)}_{x_0x_n}$ by
\begin{align}
H^{(0)}&\;:=\;\mathbb{Id}\,,\nonumber\\
H^{(1)}&\;:=\;H\,,\nonumber\\
H^{(n)}_{x_0 x_n}&\;:=\;\sum\limits_{x_1,\ldots{},x_{{n}-1}}\left(\prod_{i=0}^{n-2}\bold{1}(x_{i} \neq x_{i+2})\right) H_{x_0x_1},\ldots{},H_{x_{{n}-1}x_{n}} \hspace{3mm}(n \geq 2)\,.\nonumber\\
\end{align}
The following result is proved in \cite{1}\,.
\begin{lemma}
Let $U_k$ be the $k$-th Chebyshev polynomial of the second kind and let $$\alpha_k(t)\;\deq\; \frac{2}{\pi} \int\limits_{-1}^{1}\sqrt{1-\zeta^2}e^{-it\zeta} U_k(\zeta)d\zeta\,.$$ We define the quantity
$a_m(t)\;\deq\;\sum\limits_{k \geq 0}\frac{\alpha_{m+2k}(t)}{(M-1)^{k}}\,.$
We have that
\begin{equation}
 e^{-itH/2}\;=\;\sum\limits_{m\geq 0}a_m(t)H^{(m)}\,. 
\end{equation}

\end{lemma}
\noindent
We will use also the abbreviation 
$$\langle X ; Y \rangle \;\deq \; \mathbb{E}XY-\mathbb{E}X\mathbb{E}Y\,.$$
Plugging in the definition of $Y_T(\phi)$ we have 
\begin{align*}
\var(Y_T(\phi))\;=\;\langle Y_T(\phi) ; Y_T(\phi)\rangle\;&=\;\sum_{y_1,y_2}\phi\left(\frac{y_1}{W^{1+d\kappa/2}}\right)\phi\left(\frac{y_2}{W^{1+d\kappa/2}}\right)\langle P(t,y_1) ; P(t,y_2)\rangle\\
\;&\leq\;||\phi||^{2}_{\infty}\sum_{y_1}\sum_{y_2}|\langle P(t,y_1) ; P(t,y_2)\rangle|\,.
\end{align*}
Moreover,
\begin{align*}
\langle P(t,y_1) ; P(t,y_2)\rangle&\;=\;\nonumber\\
&\;=\;\sum\limits_{n_{11},n_{12}\geq 0}\sum\limits_{n_{21},n_{22} \geq 0}a_{n_{11}}(t)\overline{a_{n_{12}}(t)}a_{n_{21}}(t)\overline{a_{n_{22}}(t)}\langle H_{0y_1}^{(n_{11})}H_{y_10}^{(n_{12})} ; H_{0y_2}^{(n_{21})}H_{y_20}^{(n_{22})} \rangle\,.
\end{align*}
We summarize the graphical representation of  $\langle H_{0y_1}^{(n_{11})}H_{y_10}^{(n_{12})} ; H_{0y_2}^{(n_{21})}H_{y_20}^{(n_{22})} \rangle\,.$ 
\end{subsection}

\begin{subsection}{Graphical representation}

We define a graph $\mathcal{L}$ which consists of two rooted directed chains $\mathcal{L}_1$ and $\mathcal{L}_2$ by 
$$\mathcal{L}(n_{11},n_{12},n_{21},n_{22})\;\equiv\; \mathcal{L} \;\deq\;\mathcal{L}_1(n_{11},n_{12})\sqcup\mathcal{L}_2(n_{21},n_{22})\,,$$ 
where $\mathcal{L}_k(n_{k1},n_{k2})$ is a rooted directed chain of length $n_{k1}+n_{k2} \geq 1$ for $k \in \{1,2 \}.$ We denote the set of vertices of the graph $\mathcal{L}$ by $V(\mathcal{L})$ and the set of edges by $E(\mathcal{L})$. Each of the rooted directed chains contains two distinct vertices denoted by $r(\mathcal{L}_k)$ ($\textit{root}$) and 
$s(\mathcal{L}_k)$ ($\textit{summit}$) defined as the unique vertex such that the path $r(\mathcal{L}_k)\rightarrow s(\mathcal{L}_k)$ has length $n_{k1}$. Note that if $n_{k1}=0$ or $n_{k2}=0$ then $r(\mathcal{L}_k)=s(\mathcal{L}_k)$\,. Using the orientation of the edges, for each $e \in E(\mathcal{L})$ we denote the vertex $a(e) \in V(\mathcal{L})$ as predecessor and the vertex $b(e) \in V(\mathcal{L})$ as successor (see Figure $2.1$). Similarly, for each vertex $i \in V(\mathcal{L})$\,, we denote the adjacent vertices, $a(i)$ and $b(i)$ as the predecessor and the successor of $i$ (see Figure $2.2$). The root and the summit are drawn using white dots and all other vertices using black dots. Hence, the set of vertices can be split as $V(\mathcal{L})=V_{w}(\mathcal{L})\sqcup V_{b}(\mathcal{L}),$ where the subscript $w$ stands for the white vertices and $b$ for the black vertices. 

Each vertex $i \in V(\mathcal{L})$ carries a $\textit{label}$ $x_i \in \Lambda_N$\,. The labels $\bold{x}=(x_i)_{i \in V(\mathcal{L})}$ can be split according to the needs, e.g.\ 
$\bold{x}=(\bold{x}_1, \bold{x}_2),$ where $\bold{x}_k:=(\bold{x}_i)_{i \in V(\mathcal{L}_k)},\; k \in \{1,2 \},$ or $\bold{x}=(\bold{x}_b,\bold{x}_w),\; \bold{x}_b:=(x_i)_{i \in V_b(\mathcal{L})} $ and $\bold{x}_w:=(x_i)_{i \in V_w(\mathcal{L})}\,.$ 
\begin{figure}[ht]
\begin{center}
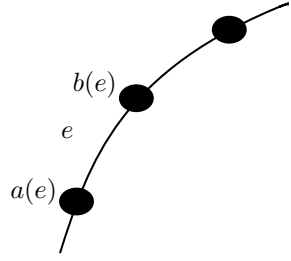
\caption{The predecessor vertex $a(e)$ and the successor vertex $b(e)$ of the edge $e$\,.}
\end{center}
\end{figure}
\begin{figure}[ht]
\begin{center}
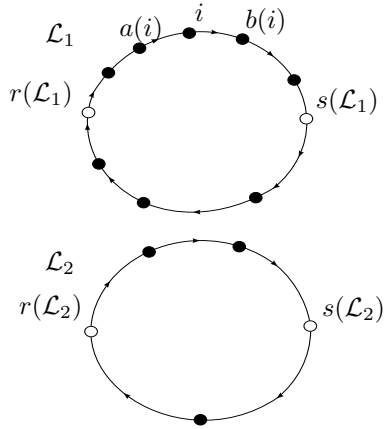
\caption{The graphical representation of the rooted directed chains.}
\end{center}
\end{figure}

For each configuration of labels $\bold{x}$ we assign a \emph{lumping} $\Gamma =\Gamma(\bold{x})$ of the set of edges $E(\mathcal{L})$ as in \cite{1}\,. A \emph{lumping} is an equivalence relation on $E(\mathcal{L})$\,. We use the notation $\Gamma=\{\gamma\}_{\gamma \in \Gamma}$ where $\gamma \in \Gamma$ is a \emph{lump}, i.e. an equivalence class of $\Gamma$\,. The lumping $\Gamma= \Gamma(\bold{x})$ associated with the labels $\bold{x}$ is given by the  equivalence relation
\begin{equation}
 e \;\sim\; e' \;\Leftrightarrow\; \{x_{a(e)}, x_{b(e)} \}\;=\;\{x_{a(e')}, x_{b(e')} \}\,.
\end{equation}
\noindent
The summation over $\bold{x}$ is performed with respect to the indicator function 
\begin{align*}
Q_{y_1,y_2}(\bold{x})\;\deq\;\bold{1}(x_{r(\mathcal{L}_1)}=0)\bold{1}(x_{r(\mathcal{L}_2)}=0)\bold{1}(x_{s(\mathcal{L}_1)}=y_1)\bold{1}(x_{s(\mathcal{L}_2)}=y_2)\prod\limits_{i \in V_b(\mathcal{L})}\bold{1}(x_{a(i)}\neq x_{b(i)})\,.
\end{align*}
Throughout the proof we will use the notation
$$x_{e}\;=\;(x_{a(e)}, x_{b(e)})\,.$$
Using the graph $\mathcal{L}$ we 
may now write the covariance as
\begin{equation}
 \langle H_{0y_1}^{(n_{11})}H_{y_10}^{(n_{12})} ; H_{0y_2}^{(n_{21})}H_{y_20}^{(n_{22})}\rangle \;=\;\sum\limits_{\bold{x} \in \Lambda_N^{V(\mathcal{L})}}Q_{y_1,y_2}(\bold{x})A(\bold{x})\,,
\end{equation}
where
\begin{align}
A(\bold{x})\;=\;\mathbb{E}\prod\limits_{e \in E(\mathcal{L})}H_{x_e}-\mathbb{E}\prod\limits_{e \in E(\mathcal{L}_1)}H_{x_e}\mathbb{E}\prod\limits_{e \in E(\mathcal{L}_2)}H_{x_e}\,.
\end{align}
\noindent
We further define the $\textit{value}$ of the lumping $\Gamma$ by
$$V_{y_1, y_2}(\Gamma)\;\deq\;\sum\limits_{\bold{x}}\bold{1}(\Gamma(\bold{x})\;=\;\Gamma)Q_{y_1,y_2}(\bold{x})A(\bold{x})\,.$$
Let $\mathfrak{P}_{c}(E(\mathcal{L}))$ be the  set of connected even lumpings, i.e. the set of all lumpings $\Gamma$ for which each lump $\gamma \in \Gamma$ has even size and there exists $\gamma \in \Gamma$ such that $\gamma \cap E(\mathcal{L}_k) \neq \emptyset\,,$ for $k \in \{1,2\}\,.$

Using that $\mathbb{E}H_{xy}=0$\,, it is not hard to see that the graphical representation of the variance yields to the following result (for further details, see \cite{3})\,.
\begin{lemma}
We have that
\begin{align*} 
\langle H_{0y_1}^{(n_{11})}H_{y_10}^{(n_{12})} ; H_{0y_2}^{(n_{21})}H_{y_20}^{(n_{22})}\rangle \;=\; \sum\limits_{\Gamma \in \mathfrak{P}_{c}(E(\mathcal{L}))}V_{y_1,y_2}(\Gamma)\,.
\end{align*}
\end{lemma}

\noindent
We define the set of all connected pairings 
$$\mathfrak{M}_c\;\deq\;\bigsqcup\limits_{n_{11}, n_{12}, n_{21}, n_{22}}\{ \Pi \in \mathfrak{P}_c(E(\mathcal{L}(n_{11}, n_{12}, n_{21}, n_{22})))\; : \; |\pi|=2\,,\hspace{2mm} \forall \pi \in \Pi\}\,.$$
We call the lumps $\pi \in \Pi$ of a pairing $\Pi$ $\textit{bridges}$.
Moreover, with each pairing $\Pi \in \mathfrak{M}_c$ we associate its underlying graph $\mathcal{L}(\Pi)$, and regard $n_{11}(\Pi)$ and $n_{12}(\Pi)$, $n_{21}(\Pi)$ and $n_{22}(\Pi)$ as functions on $\mathfrak{M}_c $ in self-explanatory notation. We abbreviate $V(\Pi)=V(\mathcal{L}(\Pi))$ and $E(\Pi)=E(\mathcal{L}(\Pi))$. We refer to $V(\Pi)$ as the set of vertices of $\Pi$ and to $E(\Pi)$ as the set of edges of $\Pi$\,.\\
Let us define the indicator function
\begin{equation}
J_{ \{ e,e' \} }(\bold{x})\;\deq\;\bold{1}(x_{a(e)}=x_{b(e')})\bold{1}(x_{a'(e)}=x_{b(e)})\,.
\end{equation}
Using the same reasoning as in Section $4$ of \cite{3} and Equation $4.14$ of \cite{3}, we obtain the following bound.
\begin{lemma}
We have\begin{align}
|\langle P(t,y_1); P(t,y_2) \rangle|
&\;\leq\;\sum\limits_{\Pi \in \mathfrak{M}_c}|a_{n_{11}(\Pi)}(t)\overline{a_{n_{12}(\Pi)}(t)}a_{n_{21}(\Pi)}(t)\overline{a_{n_{22}(\Pi)}(t)}|\sum_{\bold{x}}Q_{y_1,y_2}(\bold{x})\prod\limits_{\{e,e'\} \in \Pi}S_{x_e}\prod\limits_{\pi \in \Pi}J_{\pi}(\bold{x})\,. 
\end{align}
\end{lemma}

\begin{subsection}{Collapsing of parallel bridges}

We further construct as in \cite{3} the \emph{skeleton} $\Sigma=S(\Pi)$ of a pairing $\Pi \in \mathfrak{M}_c$ by collapsing all parallel bridges of $\Pi$\,. By definition the bridges $\{e_1, e'_1\}$ and $\{e_2, e'_2 \}$ are $\textit{parallel}$ if $b(e_1)=a(e_2)\in V_b(\Pi)$ and $b(e'_2)=a(e'_1)\in V_b(\Pi)$\,. To each $\Pi \in \mathfrak{M}_c$ we associate a couple $(\Sigma, l_{\Sigma})$, where $\Sigma \in \mathfrak{M}_c$ has no parallel bridges and $l_{\Sigma}:=(l_\sigma)_{\sigma \in \Sigma} \in \mathbb{N}^{\Sigma}$\,. The integer $l_{\sigma}$ denotes the number of parallel bridges of $\Pi$ that were collapsed into the bridge $\sigma$ of $\Sigma\,.$ Conversely, for any given couple $(\Sigma, l_{\Sigma}), $ where $\Sigma \in \mathfrak{M}_c$ has no parallel bridges and $l_\Sigma \in \mathbb{N}^{\Sigma}$, we define $\Pi=G_{l_{\Sigma}}(\Sigma)$ as the pairing obtained from $\Sigma$ by replacing for each bridge $\sigma \in \Sigma$\,, the bridge $\sigma$ with $l_{\sigma}$ parallel bridges. This construction gives a bijective mapping $\Pi \longleftrightarrow (\Sigma, l_{\Sigma})\,.$
We further define the set of admissible skeletons as 
\begin{equation}
\mathfrak{G}\;\deq\;\{S(\Pi): \Pi \in \mathfrak{M}_c \}\,.
\end{equation}
Note that all $\Sigma \in \mathfrak{G}$ are connecting.\\
The following result is to check from the definition of $\mathfrak{G}$; see Lemma $7.4   $ $(ii)$ in \cite{1}\,.
\begin{lemma}
Let $\{e, e'\} \in \Sigma$. Then $e$ and $e'$ are adjacent only if $e\cap e' \in V_w(\Sigma)\,.$
\end{lemma}
\noindent
\begin{figure}[ht]
\begin{center}
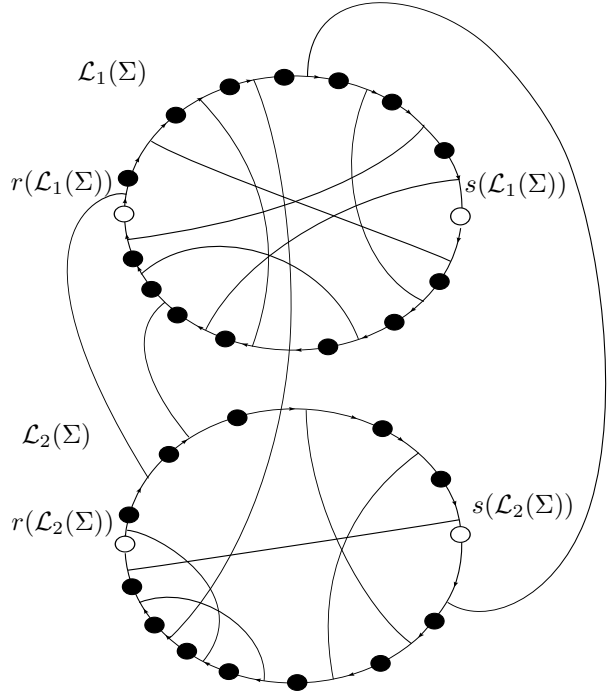
\caption{Graphical representation of the skeleton for a given configuration.}
\end{center}
\end{figure}
In the following we rewrite the right hand side of $(2.7)$ using the summation over skeleton pairings $\Sigma=S(\Pi)$, followed by different ways of expanding the bridges of $\Sigma$\,. For this, let $\Pi=G_{l_\Sigma}(\Sigma)$\,. We further define  $|l_{\Sigma}|\deq\sum_{\sigma \in \Sigma}l_{\sigma}$ for $\Sigma \in \mathfrak{G}$ and $l_{\Sigma} \in \mathbb{N}^{\Sigma}$. For the skeleton $\Sigma \in  \mathfrak{G}$ of the pairing $\Pi=G_{l_{\Sigma}}(\Sigma)$ we use the notation  $n_{ij}(\Sigma, l_{\Sigma})$ for $n_{ij}(\Pi)$, for all $i, j \in \{1,2\}$\,.\\
Parametrising $\Pi$ using $\Sigma$ and $l_{\Sigma}$ and neglecting the non-backtracking condition in the definition of $Q_{y_1,y_2}(\bold{x})$ we obtain the following upper bound (for full details see Lemma $7.6$ in \cite{1})\,.
\begin{multline}
\sum\limits_{\bold{x} \in \Lambda_N^{V(\Sigma)}}Q_{y_1, y_2}(\bold{x})\prod\limits_{\{e,e'\} \in \Sigma}\left(S^{l_{ \{ e,e'   \} } }\right)_{x_{e}}\prod\limits_{\sigma \in \Sigma}J_{\sigma}(\bold{x})
\\
\;\leq\;
\sum\limits_{\bold{x} \in \Lambda_N^{V(\Sigma)}}\bold{1}(0=x_{r(\mathcal{L}_1(\Sigma))})\bold{1}(0=x_{r(\mathcal{L}_2(\Sigma))})\bold{1}(y_1=x_{s(\mathcal{L}_1(\Sigma))})
\bold{1}(y_2=x_{s(\mathcal{L}_2(\Sigma))})\prod\limits_{\{e,e'\} \in \Sigma}\left(S^{l_{ \{ e,e'   \} } }\right)_{x_{e}}\prod\limits_{\sigma \in \Sigma}J_{\sigma}(\bold{x})\,.
\end{multline}
Let us define 
\begin{align*}
R(\Sigma)\;:=\;\sum\limits_{\bold{x}\in \Lambda_N^{V(\Sigma)}}\bold{1}(x_{r(\mathcal{L}_1(\Sigma))}=0)\bold{1}(x_{r(\mathcal{L}_2(\Sigma))}=0)\prod\limits_{\{e,e'\} \in \Sigma}\left(S^{l_{ \{ e,e'   \} } }\right)_{x_{e}}\prod\limits_{\sigma \in \Sigma}J_{\sigma}(\bold{x})\,.
\end{align*}
The following result is obtained using $(2.9)$\,.  
\begin{lemma}
We have that
\begin{align*}
\sum\limits_{y_1}\sum\limits_{y_2}\langle P(t,y_1) ; P(t,y_2)\rangle\;\leq \;\sum\limits_{\Sigma \in \mathfrak{G}}\sum\limits_{l_{\Sigma}}|a_{n_{11}(\Sigma, l_{\Sigma})}(t)\overline{a_{n_{12}(\Sigma, l_{\Sigma})}(t)}a_{n_{21}(\Sigma, l_{\Sigma})}(t)\overline{a_{n_{22}(\Sigma, l_{\Sigma})}(t)}|R(\Sigma)\,.
\end{align*}
\end{lemma}
\noindent
The following result follows easy from the definition of $S_{xy}$.
\begin{lemma}
Let $l \in \mathbb{N}$\,. For each $x,y \in \Lambda_N$ we have
\begin{enumerate}
\item $\sum\limits_{y}(S^l)_{xy}=(\frac{M}{M-1})^{l}\,.$
\item $(S^l)_{xy}\;\leq\;(\frac{M}{M-1})^{l-1}\frac{1}{M-1}\,.$
\end{enumerate}
\end{lemma}
\end{subsection}

\begin{subsection}{Orbits of vertices}

Let us fix $\Sigma \in \mathfrak{G}$\,. On the set of vertices $ V(\Sigma)$  we construct the  $\textit{orbits of vertices}$ as in \cite{1}\,. We define $\tau : V(\Sigma) \to V(\Sigma)$ as follows. Let $i \in V(\Sigma)$ and let $e$ be the unique edge such that $\{\{i, b(i)\}, e\} \in \Sigma $\,. Then, for any vertex $i$ of $\Sigma \in \mathfrak{G}$  we define
$\tau i \deq b(e)$. We denote the orbit of the vertex $i \in \Sigma$ by $[i]\;:=\; \{ \tau^n i : n \in \mathbb{N}\}$\,.\\ 
We order the edges of $\Sigma$ in some arbitrary fashion and denote this order by $<$\,. Each bridge $\sigma \in \Sigma$ ''sits between" the orbits $\zeta_1(\sigma)$ and $\zeta_2(\sigma)$. More precisely, let $\sigma=\{e , e'\}$ with $ e <e'$ and $e=\{ i, b(i)\}$\,. Then, $\zeta_1(\sigma)\deq [i]$ and $\zeta_2(\sigma)\deq [b(i)]$\,.\\
Let $Z(\Sigma):=\{[i] : i \in V(\Sigma)\}$ be the set of orbits of $\Sigma$\,. This set contains four distinguished orbits $\{ [r(\mathcal{L}_1)],[r(\mathcal{L}_2)],[s(\mathcal{L}_1)],[s(\mathcal{L}_2)] \}$ which need not be distinct. Let $|\Sigma|$ be the number of bridges of the skeleton $\Sigma \in \mathfrak{G}$ and let $L(\Sigma)=|Z^*(\Sigma)|$ with $Z^*(\Sigma):=Z(\Sigma)\setminus \{[r(\mathcal{L}_1)],[r(\mathcal{L}_2)]\}$\,.
\begin{figure}[ht]
\begin{center}
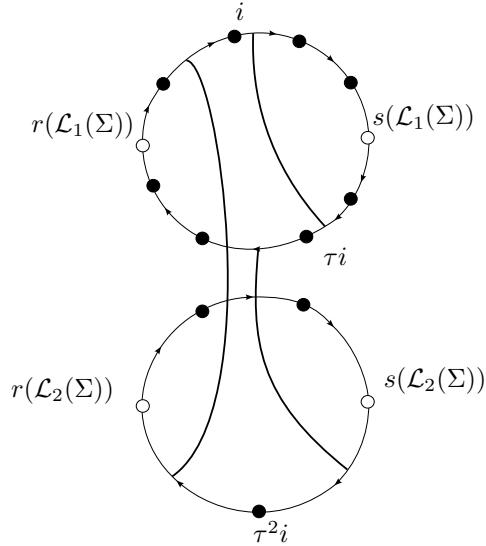
\caption{The construction of the orbit $[i]$ for the vertex $i$.}
\end{center}
\end{figure}

\noindent
The following result is an adaptation of the $\textit{2/3-rule}$ introduced in Lemma $7.7$ of \cite{1}\,.
\begin{lemma}
We have the inequality
$$ L(\Sigma )\;\leq\;\frac{2|\Sigma|}{3}+\frac{2}{3}\,.$$
\end{lemma}

\begin{proof}

Let $Z'(\Sigma):=Z(\Sigma)\setminus \{ [r(\mathcal{L}_1)], [r(\mathcal{L}_2)], [s(\mathcal{L}_1)], [s(\mathcal{L}_2)]\}$\,. Using the same reasoning as in the proof of the $\textit{2/3 rule}$ in \cite{1} we obtain that each orbit contains at least 3 vertices.

The total number of vertices of $\Sigma$  not including $\{r(\mathcal{L}_1),r(\mathcal{L}_2),s(\mathcal{L}_1),s(\mathcal{L}_2) \} $ is $2|\Sigma|-4$\,. It follows that $3|Z'(\Sigma)|\;\leq\; 2|\Sigma|-4 \Leftrightarrow |Z'(\Sigma)|\;\leq\; 2|\Sigma|/3-4/3$.
 
 Using that $|Z^*(\Sigma)|\; \leq\; |Z'(\Sigma)|+2$, we obtain $|Z^*(\Sigma)|\;\leq \; 2|\Sigma|/3+2/3$\,.
\end{proof}
\noindent
We remark that Lemma $2.7$ is sharp in the sense that there exists $\Sigma \in \mathfrak{G}$ such that the estimate of Lemma $2.7$ saturates.
\end{subsection}

\begin{section}{The case $|\Sigma|\;\geq\; 3$}

Using Lemma $2.7$ and the same argument as in Section $7.5$ of \cite{1} we obtain the following result.
\begin{lemma}
Let $\Sigma \in \mathfrak{G}$ and $l_{\Sigma} \in \Lambda_N^{V(\Sigma)}$. We have that
\begin{align*}
R(\Sigma)\;\leq\; C \left(\frac{M}{M-1}\right)^{|l_{\Sigma}|}M^{-|\Sigma|/3+ 2/3}\,.
\end{align*}
\end{lemma}

\begin{subsection}{Estimation of the variance for $|l_{\Sigma}| \ll M^{1/3}$}

Let $\mu < \frac{1}{3}$\,. In the summation $(2.12)$  we introduce a cut-off at $|l_{\Sigma}|< M^{\mu}$\,. We define
$$E^{\leq} \;\deq\; \sum\limits_{\Sigma \in \mathfrak{G}}\sum\limits_{|l_{\Sigma}|\leq M^{\mu}}|a_{n_{11}(\Sigma, l_{\Sigma})}(t)\overline{a_{n_{12}(\Sigma, l_{\Sigma})}(t)}a_{n_{21}(\Sigma, l_{\Sigma})}(t)\overline{a_{n_{22}(\Sigma, l_{\Sigma})}(t)}|R(\Sigma)\,.$$
The following result is proved in \cite{1}\,, Lemma $7.9$\,.
\begin{lemma}
\begin{enumerate}
\item For any time $t$ and for any $n \in \mathbb{N}$ we have $|a_n(t)|\leq \frac{Ct^n}{n!}\,, $ for some constant $C\,.$

\item We have $\sum\limits_{n\geq 0}|a_n(t)|^2=1+O(M^{-1})\,,$ uniformly in $t \in \mathbb{R}$\,.
\end{enumerate}
\end{lemma}
\noindent
A new estimate on $ \sum_{l_{\Sigma}}\bold{1}(|l_{\Sigma}|\;\leq\; M^{\mu})|a_{n_{11}(\Sigma, l_{\Sigma})}(t)\overline{a_{n_{12}(\Sigma, l_{\Sigma})}(t)}a_{n_{21}(\Sigma, l_{\Sigma})}(t)\overline{a_{n_{22}(\Sigma, l_{\Sigma})}(t)}|$ is established in the following lemma. The new technique is based on splitting the summation according to the bridges that touch the rooted directed chains.

\begin{lemma}
For any $\Sigma \in \mathfrak{G}$ with $|\Sigma|\geq 3$ we have 
$$ \sum\limits_{l_{\Sigma}}\bold{1}(|l_{\Sigma}|\;\leq\; M^{\mu})|a_{n_{11}(\Sigma, l_{\Sigma})}(t)\overline{a_{n_{12}(\Sigma, l_{\Sigma})}(t)}a_{n_{21}(\Sigma, l_{\Sigma})}(t)\overline{a_{n_{22}(\Sigma, l_{\Sigma})}(t)}| \leq \frac{CM^{\mu(|\Sigma|-2)}}{(|\Sigma|-3)!}\,.$$
\end{lemma}
\begin{figure}
\begin{center}
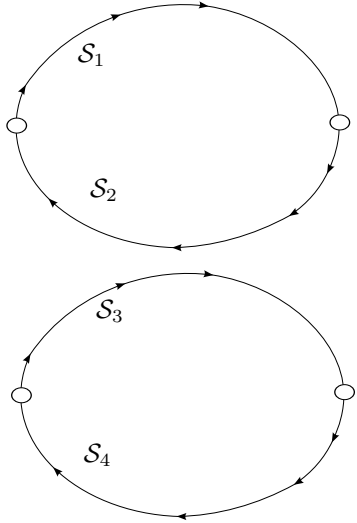
\caption{The directed paths $\cal{S}_1, \cal{S}_2, \cal{S}_3$ and $\cal{S}_4$\,.}
\end{center}
\end{figure}
\begin{proof}
Let $\Sigma \in \mathfrak{G}$\,. We denote each path by $\cal{S}_1 \equiv r(\mathcal{L}_1(\Sigma))\rightarrow s(\mathcal{L}_1(\Sigma))$, $\cal{S}_2 \equiv s(\mathcal{L}_1(\Sigma))\rightarrow r(\mathcal{L}_1(\Sigma))$, $\cal{S}_3 \equiv r(\mathcal{L}_2(\Sigma))\rightarrow s(\mathcal{L}_2(\Sigma))$ and $\cal{S}_4 \equiv s(\mathcal{L}_2(\Sigma))\rightarrow r(\mathcal{L}_2(\Sigma))$\,.
 There always exists a bridge connecting $\cal{S}_{i}$ and $\cal{S}_{j}$ , for $i \neq j\,.$ Without loss of generality we choose $\sigma_1$ connecting $\cal{S}_1$ and $\cal{S}_2$\,.\\ 
We have the following cases\,:\\
$(i)$ There is a bridge $\sigma_2 \in \Sigma$ between $\cal{S}_3$ and $\cal{S}_4$\,.\\
Let $\bar{\Sigma}\; :=\; \Sigma \setminus \{\sigma_1, \sigma_3 \}$\,. There exist the functions $f_1(l_{\bar{\Sigma}})$, $f_2(l_{\bar{\Sigma}})$, $f_3(l_{\bar{\Sigma}})$ and $f_4(l_{\bar{\Sigma}})$ such that $n_{11}(\Sigma, l_{\Sigma})=f_1(l_{\bar{\Sigma}})+l_{\sigma_1}$ and $ n_{12}(\Sigma, l_{\Sigma})=f_2(l_{\bar{\Sigma}})+l_{\sigma_1}$, $n_{21}(\Sigma, l_{\Sigma})=f_3(l_{\bar{\Sigma}})+l_{\sigma_2}$ and $ n_{22}(\Sigma, l_{\Sigma})=f_4(l_{\bar{\Sigma}})+l_{\sigma_2}$\,. Note that $n_{11}(\Sigma, l_{\Sigma})$ and $n_{21}(\Sigma, l_{\Sigma})$ do not represent the same linear combination of elements of $l_\Sigma$\,.\\
We get
\begin{multline*}
\sum\limits_{l_{\Sigma}}\bold{1}(|l_{\Sigma}|\leq M^{\mu})|a_{n_{11}(\Sigma, l_{\Sigma})}(t)\overline{a_{n_{12}(\Sigma, l_{\Sigma})}(t)}a_{n_{21}(\Sigma, l_{\Sigma})}(t)\overline{a_{n_{22}(\Sigma, l_{\Sigma})}(t)}|\\
\;=\;\sum\limits_{l_{\bar{\Sigma}}}\sum\limits_{l_{\sigma_1}}\sum\limits_{l_{\sigma_2}}\bold{1}(|l_{\Sigma}|\leq M^{\mu})|a_{f_1(l_{\bar{\Sigma}})+l_{\sigma_1}}(t)\overline{a_{f_2(l_{\bar{\Sigma}})+l_{\sigma_1}}(t)}a_{f_3(l_{\bar{\Sigma}})+l_{\sigma_2}}(t)\overline{a_{f_4(l_{\bar{\Sigma}})+l_{\sigma_2}}(t)}|\,.
\end{multline*}
Using the elementary inequality $|abcd| \leq |a|^{2}|c|^{2}+|b|^{2}|d|^{2}$ we obtain
\begin{multline}
\sum\limits_{l_{\bar{\Sigma}}}\sum\limits_{l_{\sigma_1}}\sum\limits_{l_{\sigma_2}}\bold{1}(|l_{\Sigma}|\leq M^{\mu})|a_{f_1(l_{\bar{\Sigma}})+l_{\sigma_1}}(t)\overline{a_{f_2(l_{\bar{\Sigma}})+l_{\sigma_1}}(t)}a_{f_3(l_{\bar{\Sigma}})+l_{\sigma_2}}(t)\overline{a_{f_4(l_{\bar{\Sigma}})+l_{\sigma_2}}(t)}| \\
\;\leq\; \sum\limits_{l_{\bar{\Sigma}}}\sum\limits_{l_{\sigma_1}}\sum\limits_{l_{\sigma_2}} \bold{1}(|l_{\Sigma}|\leq M^{\mu})( |a_{f_1(l_{\bar{\Sigma}})+l_{\sigma_1}}(t)|^{2}|a_{f_3(l_{\bar{\Sigma}})+l_{\sigma_2}}(t)|^2+|a_{f_2(l_{\bar{\Sigma}})+l_{\sigma_1}}(t)|^2|a_{f_4(l_{\bar{\Sigma}})+l_{\sigma_2}}(t)|^2)\,.
\end{multline}
Using the inequality between the indicator functions $\sum\limits_{\l_{\bar{\Sigma}}}\bold{1}(|l_{\Sigma}|\leq M^{\mu})\leq \sum\limits_{\l_{\bar{\Sigma}}}\bold{1}(|l_{\bar{\Sigma}}|\leq M^{\mu})$ and Lemma $3.2$ $(ii)$ we obtain that
\begin{align*}
\sum\limits_{\l_{\bar{\Sigma}}}\sum\limits_{l_{\sigma_1}}\sum\limits_{l_{\sigma_2}}\bold{1}(|l_{\Sigma}|\leq M^{\mu})|a_{f_1(l_{\bar{\Sigma}})+l_{\sigma_1}}(t)|^{2}|a_{f_3(l_{\bar{\Sigma}})+l_{\sigma_2}}(t)|^2&\;=\;\sum\limits_{{l_{\bar{\Sigma}}}}\bold{1}(|l_{\Sigma}|\leq M^{\mu})\sum\limits_{\l_{\sigma_1}}|a_{f_1(l_{\bar{\Sigma}})+l_{\sigma_1}}(t)|^{2} \sum\limits_{\l_{\sigma_2}}|a_{f_3(l_{\bar{\Sigma}})+l_{\sigma_2}}(t)|^2\nonumber\\
&\;\leq\; \sum\limits_{{l_{\bar{\Sigma}}}} \bold{1}(|l_{\bar{\Sigma}}|\leq M^{\mu})2\,.
\end{align*} 
\noindent
Note that the same argument holds for $|a_{f_2(l_{\bar{\Sigma}})+l_{\sigma_1}}(t)|^2|a_{f_4(l_{\bar{\Sigma}})+l_{\sigma_2}}(t)|^2$ \,.\\ 
Using that $|\bar{\Sigma}|=|\Sigma|-2$ we obtain that 
\begin{align}
\sum\limits_{l_1+l_2+\ldots{} +l_{|\bar{\Sigma}|}\leq M^{\mu}}2 &\;\leq\; \sum\limits_{l_2+\ldots{}+l_{|\bar{\Sigma}|}=M^{\mu}-l_1}2\nonumber\\
&\;\leq\; \sum\limits_{l_1=1}^{M^{\mu}}2{M^{\mu}-l_{1}-1 \choose{|\Sigma|-3}}\nonumber\\
&\;\leq\; C\frac{M^{\mu(|\Sigma|-2)}}{(|\Sigma|-3)!}\,.
\end{align}
$(ii)$ There is no bridge connecting $\cal{S}_3$ and $\cal{S}_4$\,.
In this case we consider two bridges $\sigma_3$ and $\sigma_4$ that are touching $\cal{S}_3$ and $\cal{S}_4$ respectively. We further define $\bar{\Sigma}\;:=\;\Sigma \setminus \{ \sigma_1, \sigma_3, \sigma_4 \}$\,.\\
We have that
\begin{multline*}
\sum\limits_{l_{{\Sigma}}}\bold{1}(|l_{\Sigma}|\leq M^{\mu})|a_{n_{11}(\Sigma, l_{\Sigma})}(t)\overline{a_{n_{12}(\Sigma, l_{\Sigma})}(t)}a_{n_{21}(\Sigma, l_{\Sigma})}(t)\overline{a_{n_{22}(\Sigma, l_{\Sigma})}(t)}|\\
\;=\;\sum\limits_{l_{\bar{\Sigma}}}\sum\limits_{l_{\sigma_1}}\sum\limits_{l_{\sigma_3}}\sum\limits_{l_{\sigma4}}\bold{1}(|l_{\Sigma}|\leq M^{\mu})|  a_{f_1(\bar{l},l_{\sigma_3},l_{\sigma_4})+l_{\sigma_1}}(t)\overline{a_{f_2(\bar{l},l_{\sigma_3},l_{\sigma_4})+l_{\sigma_1}}(t)}a_{f_3(\bar{l})+\eta_3 l_{\sigma_3}}(t)\overline{a_{f_4(\bar{l})+\eta_4 l_{\sigma_4}}(t)}|\,,
\end{multline*}
where $\eta_3\,, \eta_4 \in \{1, 2\}\,.$
Using the same inequality as in $(3.1)$ and that $l_{\bar{\sigma}}$ and $\l_{\sigma_1}$ are distinct we obtain that 
\begin{multline}
\sum\limits_{l_{\bar{\Sigma}}}\sum\limits_{l_{\sigma_1}}\sum\limits_{l_{\sigma_3}}\bold{1}(|l_{\Sigma}|\leq M^{\mu})|  a_{f_1(\bar{l},l_{\sigma_3},l_{\sigma_4})+l_{\sigma_1}}(t)|^{2}|a_{f_3(\bar{l})+\eta_3 l_{\sigma_3}}(t)|^{2}
\\
\;\leq\;
\sum\limits_{l_{\bar{\Sigma}}}\sum\limits_{l_{\sigma_1}}\bold{1}(|l_{\Sigma}|\leq M^{\mu})|  a_{f_1(\bar{l},l_{\sigma_3},l_{\sigma_4})+l_{\sigma_1}}(t)|^{2}\sum\limits_{l_{\sigma_3}}|a_{f_3(\bar{l})+\eta_3 l_{\sigma_3}}(t)|^{2}\,.
\end{multline}
The same holds for $\overline{a_{f_2(\bar{l},l_{\sigma_3},l_{\sigma_4})+l_{\sigma_1}}(t)}$ and $\overline{a_{f_4(\bar{l})+\eta_4 l_{\sigma_4}}(t)}$\,.
\\
Now the claim follows like in $(i)$\,.

\end{proof}

\noindent
Let $m=|\Sigma|$\,. Using Lemma $3.3$ and the same reasoning as in Section $7.6$ of \cite{1} we obtain
\begin{align}
E^{\leq} &\;\leq\; C \sum\limits_{m \leq M^{\mu}}\frac{M^{\mu(|\Sigma|-2)}}{(|\Sigma|-3)!}\frac{M^{2/3}}{M^{|\Sigma|/3}}\left(\frac{M}{M-1}\right)^{M^{\mu}}\nonumber\\
\end{align}
It is easy to see, as in Section $7.6$ of \cite{1}\,, that $|\{\Sigma : |\Sigma|=m\}|\;\leq\; 2^{m}m!$\,.\\
Finally, we obtain that 
\begin{align}
E^{\leq}&\;\leq\;C \sum\limits_{m=3}^{M^{\mu}}2^{m}m!\frac{M^{\mu(|\Sigma|-2)}}{(|\Sigma|-3)!}\frac{M^{2/3}}{M^{|\Sigma|/3}}\left(\frac{M}{M-1}\right)^{M^{\mu}}\nonumber\\
&\;\leq\; CM^{\mu(|\Sigma|-2)}M^{2/3-|\Sigma|/{3}}\nonumber\\
&\;\leq\;CM^{\mu-1/3}\,. 
\end{align}

\end{subsection}

\begin{subsection}{Estimation of the variance for $|l_{\Sigma}| \geq M^{1/3}$}

Let us define

\begin{equation}
E^{>} \;\deq\; \sum\limits_{\Sigma \in \mathfrak{G}}\sum\limits_{|l_{\Sigma}| \geq M^{\mu}}|a_{n_{11}(\Sigma, l_{\Sigma})}(t)\overline{a_{n_{12}(\Sigma, l_{\Sigma})}(t)}a_{n_{21}(\Sigma, l_{\Sigma})}(t)\overline{a_{n_{22}(\Sigma, l_{\Sigma})}(t)}|R(\Sigma)\,.
\end{equation}
We also define  the new set of variables\,
\begin{align} 
p_1\equiv p_1(\Sigma, l_{\Sigma})&\;=\;\frac{n_{11}(\Sigma, l_{\Sigma})+n_{12}(\Sigma, l_{\Sigma})}{2}\,, 
&p_2 \equiv p_2(\Sigma, l_{\Sigma})\;=\;\frac{n_{21}(\Sigma, l_{\Sigma})+n_{22}(\Sigma, l_{\Sigma})}{2}\,;\\
q_1 \equiv q_1(\Sigma, l_{\Sigma})&\;=\;\frac{n_{11}(\Sigma, l_{\Sigma})-n_{12}(\Sigma, l_{\Sigma})}{2}\,,
&q_2 \equiv q_2(\Sigma, l_{\Sigma})\;=\;\frac{n_{21}(\Sigma, l_{\Sigma})-n_{22}(\Sigma, l_{\Sigma})}{2}\,;
\end{align}
Note that $p_1+p_2=|l_{\Sigma}|$\,.\\
As in \cite{1}, using Lemma 3.2 $(i)$  and the inequality  $ \frac{p!}{(p-q)!} \leq \frac{(p+q)!}{p!}$ we obtain 
\begin{align}
 |a_{p_1+q_1}(t)a_{p_1-q_1}(t)a_{p_2+q_2}(t)a_{p_2-q_2}(t)|\;\leq\; C\frac{t^{2(p_1+p_2)}}{p_1!p_1!p_2!p_2!}\,.
\end{align}
Using the time scale $t \sim CM^{\kappa}$ we obtain that 
\begin{equation}
E^{>}\;\leq\;\sum\limits_{|l_{\Sigma}|>M^{\mu}}\frac{(CM^{\kappa}T)^{2(p_1+p_2)}}{p_1!p_1!p_2!p_2!}M^{2/3}\sum\limits_{m}^{p_1+p_2}\left(\frac{C(p_1+p_2)}{M^{1/3}}\right)^{m}\,.
\end{equation}
As in Section $7.7$ of \cite{1}, using that $C p_1!p_1!p_2!p_2!\geq p_1^{2p_1}p_2^{2p_2}$\,, for some constant $C$, we obtain that
\begin{align}
E^{>} &\;\leq\; \sum\limits_{|l_{\Sigma}|>M^{\mu}}M^{2/3}\left(\frac{CM^{\kappa}T}{p_1+p_2}\right)^{2(p_1+p_2)} +\sum\limits_{|l_{\Sigma}|>M^{\mu}} M^{1/3}\left(\frac{CM^{2\kappa}T^{2}}{(p_1+p_2)M^{1/3}}\right)^{p_1+p_2}\nonumber\\
&\;\leq\; \sum\limits_{|l_{\Sigma}|>M^{\mu}}M^{2/3}\left(CM^{\kappa-\mu}T\right)^{2(p_1+p_2)} +\sum\limits_{|l_{\Sigma}|>M^{\mu}} M^{1/3}\left(CM^{2\kappa-1/3-\mu}T^{2}\right)^{p_1+p_2}\nonumber\\
&\;\leq\;(CM^{\kappa-\mu+1/3M^{\mu}}T)^{2M^{\mu}}+(CM^{2\kappa-1/3-\mu + 2/3M^{\mu}}T)^{M^{\mu}}\,.
\end{align}
Choosing $\mu\;=\;1/3-\beta$ with the condition $1/3-\kappa>\mu-\kappa \;>\; 1/3M^{\mu}\;>\;1/3M^{1/3}$ (where we have $0 \;<\;\beta\;<\; 2/3-2\kappa-2/3M^{\mu} \;\leq\; 2/3-2\kappa$) completes the proof of Theorem $1$ in the case $|\Sigma|\;\geq\; 3\,.$

\begin{section}{Estimation for the variance in the case $|\Sigma|\;\leq\;2$}

\begin{subsection}{ Estimation for the variance in the case $|\Sigma|\;=\;0$ and $|\Sigma|\;=\;1$}
For $|\Sigma|=0$ we obtain that $\langle H_{00} ; H_{00}\rangle$ vanishes. Also, for $|\Sigma|=1$ we estimate term of the form
\begin{align*}
\langle\delta_{0y_1}H_{0y_1};\delta_{0y_2}H_{0y_2} \rangle\;=\;\delta_{0y_1}\delta_{0y_2}\langle H_{00};H_{00} \rangle\,.
\end{align*}
Given that $\langle H_{00}; H_{00} \rangle=0$ it follows that in the cases $|\Sigma|=0$ and $|\Sigma|=1$ the quantity of interest  is deterministic.
\begin{subsection}{ Estimation of the variance in the case $|\Sigma|\;=\;2$}
Given that the two rooted directed chains are connected we obtain that the graph with $l_{\sigma_1}$ bridges that touch $\cal{S}_1$ and $\cal{S}_2$  and $l_{\sigma_2}$ bridges that touch $\cal{S}_3$ and $\cal{S}_4$ gives no contribution to the value of the variance. Also, we obtain, up to permutations, four different possible configurations. In all four cases it holds that $y_1=y_2$\,.\\ 
We have that
\begin{multline}
\sum\limits_{l_{\sigma_1}+l_{\sigma_2}}|a_{n_{11}(\Sigma, l_{\Sigma})}(t)|^{2}|a_{n_{12}(\Sigma, l_{\Sigma})}(t)|^{2}R(\Sigma)
\\
=\;\sum\limits_{l_{\sigma_1}+l_{\sigma_2} \leq M^{\mu}}|a_{n_{11}(\Sigma, l_{\Sigma})}(t)|^{2}|a_{n_{12}(\Sigma, l_{\Sigma})}(t)|^{2}R(\Sigma)+\sum\limits_{l_{\sigma_1}+l_{\sigma_2} \geq M^{\mu}}|a_{n_{11}(\Sigma, l_{\Sigma})}(t)|^{2}|a_{n_{12}(\Sigma, l_{\Sigma})}(t)|^{2}R(\Sigma)\,.
\end{multline}
\noindent
Using Lemma $3.2$ $(ii)$ twice and Lemma $2.6$ $(ii)$ and $(i)$ we obtain that
\begin{align}
\sum\limits_{l_{\sigma_1}+l_{\sigma_2} \leq M^{\mu}}|a_{n_{11}(\Sigma, l_{\Sigma})}(t)|^{2}|a_{n_{12}(\Sigma, l_{\Sigma})}(t)|^{2}R(\Sigma)
&\;\leq\;2C \sum\limits_{l_{\sigma_1}+l_{\sigma_2}\leq M^{\mu}}\sum\limits_{\bold{x} \in \Lambda_N^{V(\Sigma)}}(S^{l_{\sigma_1}})_{x_{e}}(S^{l_{\sigma_2}})_{x_{e}}\nonumber\\
&\;\leq\; 2C\frac{1}{M-1}\left(\frac{M}{M-1}\right)^{l_{\sigma_1}-1}\sum\limits_{l_{\sigma_1}+l_{\sigma_2} \leq M^{\mu}}\sum\limits_{\bold{x} \in \Lambda_N^{V(\Sigma)}}(S^{l_{\sigma_2}})_{x_{e}}\nonumber\\
&\;\leq\;2C \frac{M^{\mu}}{M-1}\,.
\end{align}
Using again the time scale $t \sim CM^{\kappa}$ we obtain that 
\begin{multline}
\sum\limits_{l_{\sigma_1}+l_{\sigma_2} \geq M^{\mu}}|a_{n_{11}(\Sigma, l_{\Sigma})}(t)|^{2}|a_{n_{12}(\Sigma, l_{\Sigma})}(t)|^{2}R(\Sigma)
\\
\;\leq\; \sum\limits_{l_{\sigma_1}+l_{\sigma_2} \geq M^{\mu}}M^{2/3}\left(\frac{CM^{\kappa}T}{l_{\sigma_1}+l_{\sigma_2}}\right)^{2(l_{\sigma_1}+l_{\sigma_2})} +\sum\limits_{l_{\sigma_1}+l_{\sigma_2} \geq M^{\mu}} M^{1/3}\left(\frac{CM^{2\kappa}T^{2}}{(l_{\sigma_1}+l_{\sigma_2})M^{1/3}}\right)^{l_{\sigma_1}+l_{\sigma_2}}\nonumber\,.\\
\end{multline}
As in $(3.11)$ we obtain that
\begin{equation}
\sum\limits_{l_{\sigma_1}+l_{\sigma_2} \geq M^{\mu}}|a_{n_{11}(\Sigma, l_{\Sigma})}(t)|^{2}|a_{n_{12}(\Sigma, l_{\Sigma})}(t)|^{2}R(\Sigma)\;\leq\;(CM^{\kappa-\mu+1/3M^{\mu}}T)^{2M^{\mu}}+(CM^{2\kappa-1/3-\mu + 2/3M^{\mu}}T)^{M^{\mu}}\,.
\end{equation}
As before, we choose $\mu\;=\;1/3-\beta$\,. This completes the proof of Theorem $1$\,.
\end{subsection}
\end{subsection}
\end{section}

\end{subsection}
\end{section}

\end{subsection}

\end{section}

\newpage 

\bibliographystyle{plain}

\end{document}